\documentclass{amsart}
\usepackage{amsmath}
\usepackage{latexsym}
\usepackage{amsfonts}
\usepackage{amssymb}
\usepackage{mathrsfs}
\usepackage{color}
\usepackage{bbm,dsfont}
\usepackage{graphicx}
\usepackage{enumerate}
\usepackage[utf8]{inputenc}

\newtheorem{proposition}{Proposition}
\newtheorem{theorem}{Theorem}

\newtheorem{corollary}{Corollary}

\theoremstyle{definition}

\definecolor{darkgreen}{rgb}{0,0.6,0.2}

\newcommand{\Rb}{\mathbb{R}} 
\newcommand{\Cb}{\mathbb{C}} 
\newcommand{\abs}[1]{\left| #1 \right|} 

\newcommand{\pair}[2]{\langle\,#1\,,\,#2\,\rangle} 
\newcommand{\spanno}[1]{{\rm span}\left\{ #1 \right\}}
\newcommand{\no}[1]{\left\|#1\right\|} 

\newcommand{\hh}{\mathcal{H}} 
\newcommand{\lh}{\mathcal{L}(\hh)} 
\newcommand{\loneh}{\mathcal{L}^1(\hh)} 
\newcommand{\elle}[1]{\mathcal{L} (#1)}
\newcommand{\sta}{\mathcal{S}(\aa)} 
\newcommand{\ip}[2]{\left\langle\,#1\,|\,#2\,\right\rangle} 
\newcommand{\kb}[2]{|#1\rangle\langle#2|} 
\newcommand{\tr}[1]{\textrm{tr}\left[#1\right]} 
\newcommand{\id}{\mathbbm{1}} 

\renewcommand{\aa}{\mathcal{A}} 
\newcommand{\bb}{\mathcal{B}} 
\newcommand{\cc}{\mathcal{C}} 
\newcommand{\sa}[1]{#1^{\rm sa}} 
\newcommand{\votimes}{\bar{\otimes}} 
\newcommand{\ellesa}[1]{\mathcal{L}^{\rm sa} (#1)}

\newcommand{\Mo}{\mathsf{M}}
\newcommand{\Moh}{\widehat{\mathsf{M}}}
\newcommand{\No}{\mathsf{N}}
\newcommand{\Noh}{\widehat{\mathsf{N}}}
\newcommand{\Po}{\mathsf{P}}
\newcommand{\Poh}{\widehat{\mathsf{P}}}

\newcommand{\wit}{\xi} 
\newcommand{\witt}{\zeta} 
\newcommand{\De}[1]{{\mathcal{D}}(#1)} 
\newcommand{\CC}[2]{\mathcal{C}(#1;#2)}
\newcommand{\cCC}[2]{\mathcal{C}_{\rm c}(#1;#2)}
\newcommand{\CIW}[2]{\mathcal{W}(#1;#2)}
\newcommand{\bmu}{{\boldsymbol\mu}}

\newcommand{\Prg}{P^{{\rm prior}}_{{\rm guess}}}
\newcommand{\Ppg}{P^{{\rm post}}_{{\rm guess}}}
\newcommand{\en}{\mathcal{E}} 
\newcommand{\pa}{\mathscr{P}} 
\newcommand{\Phiv}{{\vec{\Phi}}} 
\newcommand{\Psiv}{{\vec{\Psi}}} 
\newcommand{\Mov}{\vec{\Mo}} 

\begin{document}

\title[]{Witnessing incompatibility of quantum channels}

\author[Carmeli]{Claudio Carmeli}
\address{\textbf{Claudio Carmeli}; D.I.M.E., Universit\`a di Genova, Savona, I-17100, Italy}
 \email{claudio.carmeli@gmail.com}

\author[Heinosaari]{Teiko Heinosaari}
\address{\textbf{Teiko Heinosaari}; Turku Centre for Quantum Physics, Department of Physics and Astronomy, University of Turku, Turku, FI-20014, Finland}
\email{teiko.heinosaari@utu.fi}

\author[Miyadera]{Takayuki Miyadera}
\address{\textbf{Takayuki Miyadera}; Department of Nuclear Engineering, University of Kyoto, Kyoto 6158540, Japan}
\email{miyadera@nucleng.kyoto-u.ac.jp}

\author[Toigo]{Alessandro Toigo}
\address{\textbf{Alessandro Toigo}; Dipartimento di Matematica, Politecnico di Milano, Milano, I-20133, Italy, and I.N.F.N., Sezione di Milano, Milano, I-20133, Italy}
\email{alessandro.toigo@polimi.it}

\begin{abstract}
We introduce the notion of incompatibility witness for quantum 
channels, defined as an affine functional that is non-negative on 
all  pairs of compatible channels and strictly negative on some 
incompatible pair. 
This notion extends the recent definition of incompatibility witnesses for quantum measurements. 
We utilize the general framework of channels acting on arbitrary finite dimensional von Neumann algebras, thus allowing us to investigate incompatibility witnesses on measurement-measurement, measurement-channel and channel-channel pairs.
We prove that any incompatibility witness can be implemented as a state discrimination 
task in which some intermediate classical information is obtained before completing the task. 
This implies that any incompatible pair of channels gives an advantage over compatible pairs in some such state discrimination task.
\end{abstract}

\maketitle

\section{Introduction}

Two input-output devices, such as measurements, channels or instruments, are called incompatible if they are not parts of a common third device \cite{HeMiZi16}.
The concept of incompatibility, taken at this level of generality, gives a common ground for several important notions and statements of quantum information.
For instance, the `no cloning' theorem is declaring that two identity channels are incompatible, and statements about optimal quantum cloning devices are then statements about compatibility of some channels, such as depolarizing channels \cite{BrDVEkFuMaSm98,Werner98,KeWe99,Cerf00bis,BrBuHi01,Hashagen17}.
As another example, antidegradable channels are exactly those channels that are compatible with themselves, whereas entanglement breaking channels are those channels that are compatible with arbitrary many copies of themselves \cite{HeMi17}.

The traditional and most extensively studied topic in the area of incompatibility is the incompatibility of pairs of measurement devices.
It has been recently shown that two quantum measurements are incompatible if and only if they give an advantage in some state discrimination task \cite{CaHeTo19,UoKrShYuGu19,SkSuCa19}. 
Physically speaking, the connection with state discrimination tasks and incompatibility of measurement devices can be understood by comparing two state discrimination scenarios, where partial information is given either before or after measurements are to be performed \cite{CaHeTo18}. 
Indeed, only for compatible pairs of measurements the state discrimination capability is unaffected by the stage when partial information is given, as pairs of this kind can be postprocessed from a single measurement device performed with no reference to partial information.
In the present paper, we show that this physical interpretation, with a slight modification, carries also to the incompatibility of quantum channels.

In \cite{CaHeTo19}, the above result was obtained by introducing the concept of incompatibility witnesses and then proving that, up to detection equivalence, every incompatibility witness is associated to some state discrimination task with partial intermediate information.
We generalize this approach and prove that a similar statement holds for all incompatible pairs of quantum channels.
To do it, we first define the concept of channel incompatibility witness and then prove that all such witnesses can be brought into a standard form related to a variation of the state discrimination task described in \cite{CaHeTo19}.
The state discrimination task we will consider does not require to couple the measured system with any ancillary system. 
In particular, the advantage of our approach is that it does not rely on entanglement.

In the formulation used in the current work, channels are completely positive linear maps between finite dimensional von Neumann algebras.
This framework (or something closely related to it) has been used in several earlier studies, e.g. \cite{Keyl02,Jencova12}.
A measurement can be seen as a channel from an abelian von Neumann algebra, and our formalism therefore covers the main theorem of \cite{CaHeTo19} as a special case.
It also allows to treat the incompatibility between channels and measurements that is the underlying source for fundamental noise-disturbance trade-off in quantum measurements \cite{HeMi13,HeReRyZi18,CaHeMiTo19,HaMi19}.

Our investigation is organized as follows. 
After having recalled some elementary facts about channels and von Neumann algebras in Sec.~\ref{sec:prel} and described the convex compact set of all compatible pairs of channels in Sec.~\ref{sec:inco}, in Sec.~\ref{sec:wit} we provide the definition and main properties of channel incompatibility witnesses. Section \ref{sec:SDT} then describes the particular state discrimination task we will be concerned with, and contains the proof that any incompatibility witness is associated to a task of this kind for some choice of the state ensemble to be detected. Finally, Secs.~\ref{sec:ex_triv} and \ref{sec:ex_clon} contain some examples of channel incompatibility witnesses. The examples of Sec.~\ref{sec:ex_triv} are derived from the measurement incompatibility witnesses constructed in \cite{CaHeTo19} by means of two mutually unbiased bases, while the example of Sec.~\ref{sec:ex_clon} is related to the optimal approximate cloning method of \cite{Werner98,KeWe99}.

\section{Preliminaries}\label{sec:prel}

We consider systems described by finite dimensional von Neumann algebras, that is, complex $*$-algebras that are isomorphic to block matrix algebras endowed with the uniform matrix norm $\no{\cdot}$. If $\aa$ is such an algebra, its {\em predual} $\aa_*$ coincides with the linear dual $\aa^*$. We denote by $\pair{a}{A}$ the canonical pairing between elements $a\in\aa_*$ and $A\in\aa$. The notations $\sa{\aa}$ and $\aa^+$ are used for the set of all selfadjoint and all positive elements of $\aa$, respectively. The analogous subsets of $\aa_*$ are
\begin{align*}
\sa{\aa}_* & = \{a\in\aa_*\mid\pair{a}{A}\in\Rb \ \ \forall A\in\sa{\aa}\}\,, \\ 
\aa^+_* &= \{a\in\aa_*\mid\pair{a}{A}\geq 0 \ \ \forall A\in\aa^+\} \,.
\end{align*}

The {\em states} of $\aa$ constitute the convex set $\sta = \{a\in\aa^+_*\mid\pair{a}{\id_\aa} = 1\}$, where $\id_\aa$ is the identity element of $\aa$. 
A {\em measurement} with a finite outcome set $X$ is described by a map $\Mo:X\to\aa$ such that $\Mo(x)\in\aa^+$ for all $x\in X$ and $\sum_{x\in X}\Mo(x) = \id_\aa$. 
The probability of obtaining an outcome $x$ by performing the measurement $\Mo$ in the state $a$ is then $\pair{a}{\Mo(x)}$. 
A measurement $\Mo$ is called {\em informationally complete} if the associated probability distributions are different for all states, i.e., for any two states $a\neq a'$ there is an outcome $x$ such that $\pair{a}{\Mo(x)} \neq \pair{a'}{\Mo(x)}$.
The informational completeness of $\Mo$ is equivalent to the condition that the real linear span of the set $\{\Mo(x)\mid x\in X\}$ coincides with the real vector space $\sa{\aa}$ \cite{SiSt92}.

A finite dimensional quantum system is associated with the von Neumann algebra $\lh$ of all linear maps on a finite dimensional complex Hilbert space $\hh$, whereas a finite classical system is described by the von Neumann algebra $\ell^\infty(X)$ of all complex functions on a finite set $X$. The respective norms are the uniform operator norm $\no{A} = \max\{\no{Au}/\no{u}\mid u\in\hh\setminus\{0\}\}$ and the sup norm $\no{F} = \max\{\abs{F(x)}\mid x\in X\}$. 
In these two extreme cases, the states of the system are described by positive trace-one operators and classical probability distributions, respectively. 
The framework of general von Neumann algebras allows us to consider also hybrid systems, like e.g.~the classical-quantum output of a quantum measuring process, or quantum systems subject to superselection rules.

Let $\bb$ be another finite dimensional von Neumann algebra. 
A {\em channel} connecting the system $\aa$ with the system $\bb$ is a linear map $\Phi:\aa_*\to\bb_*$ such that its adjoint $\Phi^*$ is completely positive and unital. 
The adjoint of $\Phi$ is the linear map $\Phi^*:\bb\to\aa$ defined by 
\begin{align*}
\pair{a}{\Phi^*(B)} = \pair{\Phi(a)}{B}
\end{align*}
for all $a\in\aa_*$ and $B\in\bb$.

A measurement $\Mo:X\to\aa$ can be regarded as a channel $\Moh:\aa_*\to\ell^1(X)$, where $\ell^1(X) = \ell^\infty(X)_*$ is the $\ell^1$-space of all complex functions on $X$. 
This identification is obtained by setting $\Moh(a) = \pair{a}{\Mo(\cdot)}$,  
or, equivalently, $\Moh^*(\delta_x) = \Mo(x)$,
where $\delta_x\in\ell^\infty(X)$ denotes the Kronecker delta function at $x$.

When $\aa=\bb=\lh$, any channel connecting the system $\aa$ with the system $\bb$ is a quantum channel in the usual sense. 
Moreover, any measurement $\Mo:X\to\lh$ is a quantum measurement in the usual sense and can be identified with a positive operator valued measure \cite{MLQT12}.
In this case, the predual $\loneh=\lh_*$ is the normed space of all linear operators on $\hh$ endowed with the trace-class norm. When instead $\aa=\bb=\ell^\infty(X)$, channels connecting $\aa$ with $\bb$ constitute classical data processings and just coincide with measurements $\Mo:X\to\ell^1(X)$ \cite{EIT06}.

\section{Incompatibility of channels}\label{sec:inco}

The incompatibility of quantum channels has been defined and studied in \cite{HeMi17,HaHePe14,Haapasalo15}. 
That definition has been generalized in \cite{Plavala17} for different types of devices in general probabilistic theories, while in \cite{Kuramochi18a} it is extended to cover the case of two channels with arbitrary outcome algebras. 
In the following we state the definition of (in)compatible channels explicitly in our current framework.

If $\bb_1$ and $\bb_2$ are two von Neumann algebras, we denote by $\bb_1\votimes\bb_2$ their algebraic tensor product canonically regarded as a von Neumann algebra; see e.g. \cite[Sec.~IV, Def.~1.3]{TOA1}. 
The projection onto the $i$th factor is the channel $\Pi_i:(\bb_1\votimes\bb_2)_*\to\bb_{i*}$ with $\Pi_1^*(B_1) = B_1\otimes\id_{\bb_2}$ and $\Pi_2^*(B_2) = \id_{\bb_1}\otimes B_2$ for all $B_i\in\bb_i$. The $i$th {\em margin} of a channel $\Phi:\aa_*\to(\bb_1\votimes\bb_2)_*$ is then defined as the composition channel $\Pi_i\circ\Phi$. 
Two channels $\Phi_1:\aa_*\to\bb_{1*}$ and $\Phi_2:\aa_*\to\bb_{2*}$ are {\em compatible} if there exists a channel $\Phi$ such that $\Phi_1 = \Pi_1\circ \Phi$ and $\Phi_2 = \Pi_2\circ \Phi$. 
In this case, we say that $\Phi$ is a {\em joint channel} of $\Phi_1$ and $\Phi_2$. 
Otherwise, $\Phi_1$ and $\Phi_2$ are called {\em incompatible}.

The compatibility of $\Phi_1$ and $\Phi_2$ is preserved if they are concatenated with other channels $\Psi_i : \bb_{i*}\to\cc_{i*}$. 
Indeed, if $\Phi$ is a joint channel of $\Phi_1$ and $\Phi_2$, then the composition $(\Psi_1\otimes\Psi_2)\circ\Phi$ is a joint channel of $\Psi_1\circ\Phi_1$ and $\Psi_2\circ\Phi_2$.
The tensor product of two channels is defined by the relation $(\Psi_1\otimes\Psi_2)^* = \Psi_1^*\otimes\Psi_2^*$.

In the particular case $\bb_1 = \ell^\infty(X_1)$ and $\bb_2 = \ell^\infty(X_2)$, the compatibility of channels coincides with the usual notion of compatibility for measurements due to the aforementioned identification $\Mo\simeq\Moh$ \cite[Prop.~5]{HeMi17}. 
Indeed, let $\pi_i$ be the projection onto the $i$th factor of the Cartesian product $X_1\times X_2$, and recall that two measurements $\Mo_1 : X_1\to\aa$ and $\Mo_2 : X_2\to\aa$ are called {\em compatible} if there exists a third measurement $\Mo : X_1\times X_2\to\aa$ such that its {\em margins} $\pi_i\Mo(x_i) = \sum_{(y_1,y_2)\in \pi_i^{-1}(x_i)} \Mo(y_1,y_2)$ coincide with $\Mo_i$. 
The equivalence of the two notions of compatibility then directly follows from the equality $\Pi_i\circ\Moh = \widehat{\pi_i\Mo}$.

Similarly, when $\bb_1=\ell^\infty(X)$ and $\bb_2=\lh$, any channel connecting the system $\aa$ with the system $\bb_1\votimes\bb_2$ can be identified with an {\em instrument} \cite{MLQT12}. 
In this case, compatibility of two channels $\Phi_1:\aa_*\to\ell^1(X)$ and $\Phi_2:\aa_*\to\loneh$ amounts to measurement-channel compatibility in the sense of \cite{HeMi13}.

We denote by $\CC{\aa}{\bb_1,\bb_2}$ the convex compact set of all pairs of channels $(\Phi_1,\Phi_2)$, where $\Phi_i : \aa_*\to\bb_{i*}$. 
Convex combinations in $\CC{\aa}{\bb_1,\bb_2}$ are defined componentwise.
We let $\cCC{\aa}{\bb_1,\bb_2}$ be the subset of all compatible pairs of channels. 
This subset is itself convex and compact, since it is the image of the convex compact set of channels $\Phi:\aa_*\to(\bb_1\votimes\bb_2)_*$ under the convex mapping $\Phi\mapsto (\Pi_1\circ\Phi,\Pi_2\circ\Phi)$. 
In the following, we show that the inclusion $\cCC{\aa}{\bb_1,\bb_2}\subseteq\CC{\aa}{\bb_1,\bb_2}$ is strict unless $\aa$ is a commutative algebra or either $\bb_1$ or $\bb_2$ is trivial.
A different but related result has been proven in \cite{Plavala16}.

\begin{proposition}\label{prop:triviality}
$\cCC{\aa}{\bb_1,\bb_2}=\CC{\aa}{\bb_1,\bb_2}$ if and only if $\aa$ is an abelian von Neumann algebra or $\bb_i = \Cb$ for $i\in\{1,2\}$.
\end{proposition}

\begin{proof}
If $\bb_i = \Cb$ for -- say -- $i=1$, then the trivial channel $\Phi_1 = \pair{\cdot}{\id_\aa}$ is the unique channel connecting the system $\aa$ with the system $\bb_1$. This channel is compatible with any channel $\Phi_2 : \aa_*\to\bb_{2*}$. Indeed, $\Phi_2$ is itself a joint channel of $\Phi_1$ and $\Phi_2$ since $\bb_1\votimes\bb_2 = \bb_2$. 
Thus, $\cCC{\aa}{\bb_1,\bb_2}=\CC{\aa}{\bb_1,\bb_2}$ in this case.

If $\aa$ is abelian, there exists a finite set $X$ such that $\aa$ is isomorphic to the von Neumann algebra $\ell^\infty(X)$. Since $\ell^\infty(X)_* = \ell^1(X)$ and $(\ell^\infty(X)\votimes\ell^\infty(X))_* = \ell^1(X\times X)$, we can define a broadcasting map $\Gamma:\aa_*\to (\aa\votimes\aa)_*$ as
$$
[\Gamma(f)](x,y) = \begin{cases}
f(x) & \text{ if $x=y$} \,, \\
0 & \text{ if $x\neq y$} \,.
\end{cases}
$$
The adjoint $\Gamma^*$ is positive and unital, hence $\Gamma$ is a channel. Indeed, for any linear map having an abelian von Neumann algebra as its domain or image, positivity implies complete positivity by \cite[Thms.~3.9 and 3.11]{CBMOA03}. 
The two margins of $\Gamma$ are the identity channel ${\rm id} : \aa_*\to\aa_*$. 
Therefore, the fact that compatibility is preserved in concatenation implies that any two channels $\Phi_1=\Phi_1\circ{\rm id}$ and $\Phi_2=\Phi_2\circ{\rm id}$ are compatible.
We conclude that $\cCC{\aa}{\bb_1,\bb_2}=\CC{\aa}{\bb_1,\bb_2}$ also in this case.

Finally, if $\cCC{\aa}{\bb_1,\bb_2}=\CC{\aa}{\bb_1,\bb_2}$, then either $\bb_i = \Cb$ for some $i\in\{1,2\}$, or for all $i\in\{1,2\}$ there exist two disjointly supported states $b_{i,1},b_{i,2}\in\mathcal{S}(\bb_i)$. In the latter case, either $\aa=\Cb$ and thus $\aa$ is abelian, or, for any two fixed projections $P_{1,1},P_{2,1}\in\aa\setminus\{0,\id_\aa\}$, let $P_{1,2},P_{2,2}\in\aa$ be such that $P_{i,1} + P_{i,2} = \id_\aa$ for $i=1,2$. Further, let $\Phi_i:\aa_*\to\bb_{i*}$ be the linear map defined as
\begin{equation*}
\Phi_i(a) = \sum_{k=1,2} \pair{a}{P_{i,k}}\, b_{i,k} \,.
\end{equation*}
The unitality of $\Phi_i^*$ is clear. Moreover, since $\Phi_i^*$ is positive and its image $\Phi_i^*(\bb_i) = \spanno{P_{i,1},P_{i,2}}$ is a commutative algebra, it follows that $\Phi_i^*$ is completely positive by \cite[Thm.~3.9]{CBMOA03}. Thus, $\Phi_i$ is a channel, and by the assumed hypothesis we can pick a joint channel $\Phi$ of $\Phi_1$ and $\Phi_2$. If $Q_{i,1}$ and $Q_{i,2}$ are the support projections of the states $b_{i,1}$ and $b_{i,2}$, respectively, then
\begin{equation*}
0\leq \Phi^*(Q_{1,h}\otimes Q_{2,k}) \leq
\begin{cases}
\Phi^*(Q_{1,h}\otimes \id_{\bb_2}) = \Phi_1^*(Q_{1,h}) = P_{1,h} \,,\\
\Phi^*(\id_{\bb_1}\otimes Q_{2,k}) = \Phi_2^*(Q_{2,k}) = P_{2,k}
\end{cases}
\end{equation*}
for all $h,k=1,2$. It follows that $\Phi^*(Q_{1,h}\otimes Q_{2,k})\Phi^*(Q_{1,h'}\otimes Q_{2,k'}) = 0$ whenever $(h,k)\neq (h',k')$, and hence the projections
$$
P_{1,1} = \Phi^*(Q_{1,1}\otimes Q_{2,1}) + \Phi^*(Q_{1,1}\otimes Q_{2,2})\,, \quad P_{2,1} = \Phi^*(Q_{1,1}\otimes Q_{2,1}) + \Phi^*(Q_{1,2}\otimes Q_{2,1})
$$
commute. Since the choice of $P_{1,1}$ and $P_{2,1}$ was arbitrary, this proves that all projections commute in $\aa$, which again implies that $\aa$ is abelian.
\end{proof}

The next corollary is a restatement of \cite[Thm.~3]{BaBaLeWi07} within the framework of von Neumann algebras. 
Interestingly, the assumption that $\aa$ is finite dimensional is essential for its validity \cite[Thm.~3.10]{KaLuLu15}.

\begin{corollary}
The identity channel ${\rm id}:\aa_*\to\aa_*$ is compatible with itself if and only if $\aa$ is abelian.
\end{corollary}

\begin{proof}
Since any channel $\Phi:\aa_*\to\aa_*$ is the composition $\Phi=\Phi\circ{\rm id}$, the inclusion $({\rm id},{\rm id})\in\cCC{\aa}{\aa,\aa}$ is equivalent to the equality $\cCC{\aa}{\aa,\aa} = \CC{\aa}{\aa,\aa}$, and then to $\aa$ being abelian by Proposition \ref{prop:triviality}.
\end{proof}

\section{Channel incompatibility witnesses}\label{sec:wit}

From now on, we will always assume that the inclusion $\cCC{\aa}{\bb_1,\bb_2}\subseteq\CC{\aa}{\bb_1,\bb_2}$ is strict. In wiew of Proposition \ref{prop:triviality}, this amounts to require that $\aa$ is not abelian and $\dim\bb_i \geq 2$ for all $i=1,2$.

For convenience, we denote $\Phiv=(\Phi_1,\Phi_2)$.
A {\em (channel) incompatibility witness} (CIW) is a map $\wit:\CC{\aa}{\bb_1,\bb_2}\to\Rb$ having the following three properties:
\begin{enumerate}[(W1)]
\item $\wit(\Phiv) \geq 0$ for all $\Phiv\in\cCC{\aa}{\bb_1,\bb_2}$;\label{it:CIWpositivity}
\item $\wit(\Phiv) < 0$ at least for some incompatible pair $\Phiv\in\CC{\aa}{\bb_1,\bb_2}$;\label{it:CIWnegativity}
\item $\wit(t\Phiv + (1-t)\Psiv))=t\wit(\Phiv) + (1-t)\wit(\Psiv)$ for all $\Phiv,\Psiv\in\CC{\aa}{\bb_1,\bb_2}$ and $t\in(0,1)$.\label{it:CIWaffinity}
\end{enumerate}
We denote by $\CIW{\aa}{\bb_1,\bb_2}$ the set of all such maps $\wit$.

If $\wit\in\CIW{\aa}{\bb_1,\bb_2}$ and $\wit(\Phiv) < 0$, we say that $\wit$ {\em detects} the incompatible pair of channels $\Phiv$; the set of all detected pairs is denoted by $\De{\wit}$. 
The larger is the set $\De{\wit}$, the more efficient is the CIW $\wit$ in detecting incompatibility. 
Given another $\wit'\in\CIW{\aa}{\bb_1,\bb_2}$, we say that $\wit'$ is {\em finer} than $\wit$ whenever $\De{\wit}\subseteq\De{\wit'}$. 
Further, two witnesses $\wit$ and $\wit'$ are called {\em detection equivalent} if $\De{\wit} = \De{\wit'}$. 
For any choice of $\wit$, we can always construct another CIW $\bar{\wit}$ which is finer than $\wit$ by setting 
\begin{equation}\label{eq:wit'}
\bar{\wit}(\Phiv) = \wit(\Phiv) - \min\{\wit(\Psiv)\mid\Psiv\in\cCC{\aa}{\bb_1,\bb_2}\} \, .
\end{equation}
In the case $\wit = \bar{\wit}$, we say that $\wit$ is {\em tight}.

We observe that when restricting to the particular case in which $\aa$ is a full matrix algebra and the algebras $\bb_i$ are abelian, the above definition of CIW coincides with the definition of incompatibility witnesses for quantum measurements introduced in \cite{CaHeTo19}. Indeed, as we have already seen, measurements $\Mo_i : X_i\to\aa$ and channels $\Phi_i:\aa_*\to\bb_{i*}$ are naturally identified when $\bb_i = \ell^\infty(X_i)$, and the two notions of compatibility for measurements and channels are the same under this identification. 
Properties (W\ref{it:CIWpositivity})-(W\ref{it:CIWaffinity}) are then a rewriting of the similar ones stated in \cite{CaHeTo19}. 
Related investigations on incompatibility witnesses have been reported in \cite{Jencova18,BlNe18}.

By standard separation results for convex compact sets, witnesses are enough to detect all incompatible pairs of channels.

\begin{proposition}\label{prop:separation}
For any incompatible pair of channels $\Phi_1:\aa_*\to\bb_{1*}$ and $\Phi_2:\aa_*\to\bb_{2*}$, there exist a channel incompatibility witness $\wit\in\CIW{\aa}{\bb_1,\bb_2}$ detecting the pair $(\Phi_1,\Phi_2)$.
\end{proposition}

\begin{proof}
Denote by $\ellesa{\aa_*;\bb_{i*}}$ the real vector space of all complex linear maps $\Phi_i:\aa_*\to\bb_{i*}$ satisfying $\Phi_{i*}(\sa{\aa}_*)\subseteq\sa{\bb}_{i*}$. Then, the sets $\CC{\aa}{\bb_1,\bb_2}$ and $\cCC{\aa}{\bb_1,\bb_2}$ are convex compact subset of the Cartesian product $\ellesa{\aa_*;\bb_{1*}}\times\ellesa{\aa_*;\bb_{2*}}$. 
If $\Phiv\notin\cCC{\aa}{\bb_1,\bb_2}$, by \cite[Cor.~11.4.2]{CA70} there exist elements $\phi_i\in\ellesa{\aa_*;\bb_{i*}}^*$ and $\delta\in\Rb$ such that $\sum_{i=1,2}\pair{\phi_i}{\Phi_i} > \delta$ and $\sum_{i=1,2}\pair{\phi_i}{\Psi_i} \leq \delta$ for all $\Psiv\in\cCC{\aa}{\bb_1,\bb_2}$. Here, $\ellesa{\aa_*;\bb_{i*}}^*$ denotes the linear dual of $\ellesa{\aa_*;\bb_{i*}}$, and $\pair{\phi_i}{\Phi_i}$ is the canonical pairing between elements $\phi_i\in\ellesa{\aa_*;\bb_{i*}}^*$ and $\Phi_i\in\ellesa{\aa_*;\bb_{i*}}$. Setting $\wit(\Psiv) = \delta - \sum_{i=1,2}\pair{\phi_i}{\Psi_i}$ for all $\Psiv\in\CC{\aa}{\bb_1,\bb_2}$, we thus obtain a CIW for which $\Phiv\in\De{\wit}$.
\end{proof}

\section{Channel incompatibility witnesses as a state discrimination task}\label{sec:SDT}

We consider the following state discrimination task, in which Bob is asked to retrieve a string of classical information which Alice sends to him through some communication channel which can be classical, quantum or semi-quantum.

\begin{enumerate}[(i)]

\item Alice randomly picks a label $z$ with probability $p(z)$ and she encodes it into a state $a_z$. 
The label $z$ is chosen within either one of two finite disjoint sets $X_1$ and $X_2$. The state $a_z$ belongs to the predual of the von Neumann agebra $\aa$ which describes Alice's system.\label{it:task_begins}

\item Alice then sends the state $a_z$ to Bob. At a later and still unspecified time, she also communicates him the set $X_i$ from which she picked the label $z$.

\item  Bob processes the received state $a_z$ by converting it into a bipartite system $\bb = \bb_1\votimes\bb_2$. This amounts to applying a channel $\Phi:\aa_*\to (\bb_1\votimes\bb_2)_*$, thus obtaining the bipartite state $\Phi(a_z)$ on Bob's side. 

\item In order to retrieve the label $z$, Bob performs two local measurements $\Mo_1$ on the subsystem $\bb_1$ and $\Mo_2$ on the subsystem $\bb_2$. Each measurement $\Mo_j$ has outcomes in the corresponding label set $X_j$. The probability that Bob jointly obtains the outcomes $x_1$ and $x_2$ from the respective measurements $\Mo_1$ and $\Mo_2$ is thus $\pair{\Phi(a_z)}{\Mo_1(x_1)\otimes\Mo_2(x_2)}$.

\item Finally, according to the set $X_i$ communicated by Alice, Bob's guess for the label $z$ is the outcome $x_i$.\label{it:task_ends}
\end{enumerate}

The disjoint sets $X_1$ and $X_2$, the probability $p$ on the union $X_1\cup X_2$ and the states $\{a_z\mid z\in X_1\cup X_2\}$ used by Alice in her encoding are fixed and known by both parties.
Also the two measurements $\Mo_1$ and $\Mo_2$ used by Bob are fixed. 
Only the channel $\Phi:\aa_*\to (\bb_1\votimes\bb_2)_*$ can be freely chosen by Bob.

According to the time when Alice communicates to Bob the chosen set $X_i$, two scenarios then arise.
\begin{enumerate}[(a)]
\item {\em Preprocessing information scenario}: Alice communicates the value of $i$ to Bob {\em before} he processes the received state $a_z$. Bob can then optimize the choice of $\Phi$ according to Alice's information. If $\Phi_{(i)}$ is the channel he uses when Alice communicates him the set $X_i$, his probability of guessing the correct label is
\begin{align*}
\Prg & = \sum_{i=1,2}\,\sum_{\substack{x_1\in X_1\\ x_2\in X_2}} p(x_i)\pair{\Phi_{(i)}(a_{x_i})}{\Mo_1(x_1)\otimes\Mo_2(x_2)} \\
&= \sum_{i=1,2}\,\sum_{z\in X_i} p(z)\pair{\Pi_i\circ\Phi_{(i)}(a_z)}{\Mo_i(z)} \,.
\end{align*}
This quantity depends only on the two margin channels $\Phi_1 = \Pi_1\circ\Phi_{(1)}$ and $\Phi_2 = \Pi_2\circ\Phi_{(2)}$. Since $\Phi_{(1)}$ and $\Phi_{(2)}$ are arbitrary, the pair $(\Phi_1,\Phi_2)$ can be any element of $\CC{\aa}{\bb_1,\bb_2}$.\label{it:pre_scenario}
\item {\em Postprocessing information scenario}: Alice communicates the value of $i$ to Bob {\em after} he processes the received state $a_z$. 
Bob is then forced to choose $\Phi$ without knowing the set $X_i$ chosen by Alice.
His channel $\Phi$ is thus the same regardless of the value of $i$. 
In this scenario, Bob's probability of guessing the correct label is
\begin{align*}
\Ppg & = \sum_{i=1,2}\,\sum_{\substack{x_1\in X_1\\ x_2\in X_2}} p(x_i)\pair{\Phi(a_{x_i})}{\Mo_1(x_1)\otimes\Mo_2(x_2)} \\
& = \sum_{i=1,2} \, \sum_{z\in X_i} p(z)\pair{\Pi_i\circ\Phi(a_z)}{\Mo_i(z)} \,.
\end{align*}
The latter quantity depends on the two margins $\Phi_1 = \Pi_1\circ\Phi$ and $\Phi_2 = \Pi_2\circ\Phi$ of a single channel $\Phi$.
These need to be a pair of compatible channels $(\Phi_1,\Phi_2)\in\cCC{\aa}{\bb_1,\bb_2}$.\label{it:post_scenario}
\end{enumerate}

It is useful to merge the probability distribution $p$ on $X_1\cup X_2$ and the states $\{a_z\mid z\in X_1\cup X_2\}$ into a single map $\en:X_1\cup X_2\to\aa_*$, defined as $\en(z) = p(z)\,a_z$. We call this map a {\em state ensemble} with label set $X_1\cup X_2$.  Its defining properties are that $\en(z)\in\aa^+_*$ for all $z$ and $\sum_{z\in X_1\cup X_2} \en(z)\in\mathcal{S}(\aa)$.
We further denote by $\pa$ the pair of disjoint sets $(X_1,X_2)$, and we collect the two measurements $\Mo_1$ and $\Mo_2$ within a single vector $\Mov = (\Mo_1,\Mo_2)$. 
The procedure described in steps \eqref{it:task_begins}-\eqref{it:task_ends} is thus completely determined by the triple $(\pa,\en,\Mov)$, together with the choice between scenarios \eqref{it:pre_scenario} and \eqref{it:post_scenario}.

In the two guessing probabilities described above, the pair of sets $\pa$, the state ensemble $\en$ and the measurement vector $\Mov$ are fixed parameters, while the channels $\Phi_i = \Pi_i\circ \Phi_{(i)}$ and $\Phi$ are variable quantities. To stress it, we rewrite
\begin{align}
& \Prg(\Phi_1,\Phi_2 \parallel \pa,\en,\Mov) = \sum_{i=1,2}\sum_{z\in X_i} \pair{\Phi_i(\en(z))}{\Mo_i(z)} \,, \\ 
& \Ppg(\Phi\parallel\pa,\en,\Mov) = \Prg(\Pi_1\circ\Phi,\Pi_2\circ\Phi \parallel \pa,\en,\Mov) \,. 
\end{align}
Optimizing these probabilities over the respective sets of channels, we obtain Bob's maximal guessing probabilities in the two scenarios:
\begin{gather}
\Prg(\pa,\en,\Mov) = \max\{\Prg(\Phiv \parallel \pa,\en,\Mov)\mid \Phiv\in\CC{\aa}{\bb_1,\bb_2}\} \,, \label{eq:Pprior}\\
\Ppg(\pa,\en,\Mov) = \max\{\Prg(\Phiv \parallel \pa,\en,\Mov)\mid \Phiv\in\cCC{\aa}{\bb_1,\bb_2}\} \,. \label{eq:Ppost}
\end{gather}
Clearly, $\Prg(\pa,\en,\Mov) \geq \Ppg(\pa,\en,\Mov)$. 
Whenever the inequality is strict, the expression
\begin{equation}\label{eq:associated_CIW}
\wit_{\pa,\en,\Mov}(\Phiv) = \Ppg(\pa,\en,\Mov) - \Prg(\Phiv \parallel \pa,\en,\Mov) \qquad \forall \Phiv\in\CC{\aa}{\bb_1,\bb_2}
\end{equation}
defines a tight CIW $\wit_{\pa,\en,\Mov}\in\CIW{\aa}{\bb_1,\bb_2}$. We call it the CIW {\em associated with the state discrimination task} $(\pa,\en,\Mov)$. 
Remarkably, no generality is lost in considering only CIWs of this form, as it is shown in the following main result.

\begin{theorem}\label{th:main}
Suppose $X_1$ and $X_2$ are two finite disjoint sets, $\Mo_1:X_1\to\bb_1$ and $\Mo_2:X_2\to\bb_2$ are two informationally complete measurements, and let $\pa = (X_1,X_2)$ and $\Mov = (\Mo_1,\Mo_2)$. Then, for any channel incompatibility witness $\wit\in\CIW{\aa}{\bb_1,\bb_2}$, there exists a state ensemble $\en:X_1\cup X_2\to\aa_*$ and real constants $\alpha > 0$ and $\Ppg(\pa,\en,\Mov)\leq\delta<\Prg(\pa,\en,\Mov)$ such that
\begin{equation}\label{eq:standard}
\wit(\Phiv) = \alpha\big[\delta - \Prg(\Phiv \parallel \pa,\en,\Mov)\big] \qquad \forall \Phiv\in\CC{\aa}{\bb_1,\bb_2}\,.
\end{equation}
In particular, $\bar\wit = \alpha\wit_{\pa,\en,\Mov}$, and thus the channel incompatibility witness $\wit_{\pa,\en,\Mov}$ is finer than $\wit$.
\end{theorem}

We emphasize that in Theorem \ref{th:main} the sets $X_1$, $X_2$ and the measurements $\Mo_1$, $\Mo_2$ are fixed quantities, while the state ensemble $\en$ and the real constants $\alpha$ and $\delta$ depend upon the CIW at hand. Thus, the only free parameters which  effectively enter the description of an arbitrary CIWs are just the quantities $\en$, $\alpha$ and $\delta$. For a tight CIW, the free parameters actually reduce to only $\en$ and $\alpha$.
We further note that by \cite[Prop.~1]{HeMaWo13}, there exist informationally complete measurements $\Mo_1:X_1\to\bb_1$ and $\Mo_2:X_2\to\bb_2$ such that the cardinalities of the respective outcome sets are $\abs{X_i} = \dim\bb_i$. 
As a consequence of this fact, one can always choose $X_1,X_2$ with cardinalities $\abs{X_i} = \dim\bb_i$.

\begin{proof}[Proof of Theorem \ref{th:main}]
As we have already seen in the proof of Proposition \ref{prop:separation}, the set $\CC{\aa}{\bb_1,\bb_2}$ is a convex subset of the Cartesian product $\ellesa{\aa_*;\bb_{1*}}\times\ellesa{\aa_*;\bb_{2*}}$, where we denote by $\ellesa{\aa_*;\bb_{i*}}$ the real vector space of all complex linear maps $\Phi_i:\aa_*\to\bb_{i*}$ satisfying $\Phi_{i*}(\sa{\aa}_*)\subseteq\sa{\bb}_{i*}$. Then, for any CIW $\wit\in\CIW{\aa}{\bb_1,\bb_2}$, by \cite[Prop.~S2 of the Supplementary Material]{CaHeTo19} there exist a dual element $(\phi_1,\phi_2)\in\ellesa{\aa_*;\bb_{1*}}^*\times\elle{\aa_*;\bb_{2*}}^*$ and $\delta_0\in\Rb$ such that
\begin{equation*}
\wit(\Phiv) = \delta_0 - \sum_{i=1,2} \pair{\phi_i}{\Phi_i} \qquad \forall \Phiv\in\CC{\aa}{\bb_1,\bb_2}\,.
\end{equation*}
The dual space $\ellesa{\aa_*;\bb_{i*}}^*$ is identified with the real algebraic tensor product $\sa{\aa}_*\otimes\sa{\bb}_i$ by setting
\begin{equation*}
\pair{a\otimes B_i}{\Phi_i} = \pair{\Phi_i (a)}{B_i} \qquad \forall a\in\sa{\aa}_*,\,B_i\in\sa{\bb}_i,\,\Phi_i\in\ellesa{\aa_*;\bb_{i*}} \,.
\end{equation*}
Then, since the set $\{\Mo_i(z)\mid z\in X_i\}$ spans $\sa{\bb}_i$, we have
\begin{equation*}
\phi_i = \sum_{z\in X_i} a_i(z)\otimes\Mo_i(z)
\end{equation*}
for some choice of elements $\{a_i(z)\mid z\in X_i\}$. Now, fix any faithful state of $\aa$, that is, any $a_0\in\mathcal{S}(\aa)$ such that $\pair{a_0}{A} > 0$ for all $A\in\aa^+$ with $A\neq 0$. Such a state exists by standard arguments \cite[Sec.~I.9, Exercise 3.(b)]{TOA1}. Then, if $\beta\in\Rb$ is such that
\begin{equation*}
\beta > \frac{\max \{\no{a_i(z)} \mid z\in X_i \,, i=1,2\}}{\min \{\pair{a_0}{A} \mid A\in\aa^+,\, \no{A} = 1\}} \,,
\end{equation*}
we have $\pair{\beta a_0 + a_i(z)}{A} > 0$ for all $A\in\aa^+$ with $A\neq 0$ and $z\in X_i$, $i=1,2$. Therefore, we can define the state ensemble $\en : X_1\cup X_2\to\aa_*$ given by
\begin{equation*}
\en(z) = \frac{1}{\alpha}\left(\beta a_0 + a_i(z)\right) \qquad\forall z\in X_i, \, i=1,2 \,,
\end{equation*}
where the normalization constant $\alpha>0$ is
\begin{equation*}
\alpha = \sum_{i=1,2} \sum_{z\in X_i} \pair{\beta a_0 + a_i(z)}{\id_\aa} \,.
\end{equation*}
{For the state ensemble $\en$, we have
\begin{align*}
\wit(\Phiv) & = \delta_0 - \sum_{i=1,2} \sum_{z\in X_i}\pair{a_i(z)\otimes\Mo_i(z)}{\Phi_i} = \delta_0 - \sum_{i=1,2} \sum_{z\in X_i}\pair{\Phi_i(a_i(z))}{\Mo_i(z)} \\
& = \delta_0 + 2\beta - \alpha\sum_{i=1,2} \sum_{z\in X_i}\pair{\Phi_i(\en(z))}{\Mo_i(z)} \\
& = \alpha\big[\delta - \Prg(\Phiv\parallel\pa,\en,\Mov)\big]\,,
\end{align*}
in which we set $\delta = (\delta_0 + 2\beta)/\alpha$.
Since $\wit$ is a CIW, property (W\ref{it:CIWpositivity}) and \eqref{eq:Ppost} imply the inequality $\delta\geq \Ppg(\pa,\en,\Mov)$, while on the other hand property (W\ref{it:CIWnegativity}) and \eqref{eq:Pprior} require that $\delta < \Prg(\pa,\en,\Mov)$. By inserting \eqref{eq:standard} into \eqref{eq:wit'} and using again \eqref{eq:Ppost}, we immediately obtain the equality $\bar\wit = \alpha\wit_{\pa,\en,\Mov}$, and hence the CIW $\wit_{\pa,\en,\Mov}$ is finer than $\wit$.}
\end{proof}

As a consequence of Theorem \ref{th:main}, for any pair of incompatible channels $(\Phi_1,\Phi_2)\in\CC{\aa}{\bb_1,\bb_2}$, there exists some state discrimination task in which Bob can improve his guessing probability by choosing among $\Phi_1$ and $\Phi_2$ according to the preprocessing information. 
From an equivalent point of view, whenever Bob's strategy is to arrange his channel $\Phi_i$ after he knows the value of $i$, one can find a triple $(\pa,\en,\Mov)$ that reveals Bob's use of preprocessing information. 
This is the content of the next corollary.

\begin{corollary}\label{cor:main}
Let $\pa = (X_1,X_2)$ and $\Mov = (\Mo_1,\Mo_2)$, with $X_i$ and $\Mo_i$ as in Theorem \ref{th:main}. 
Two channels $\Phi_1 : \aa_*\to\bb_{1*}$ and $\Phi_2 : \aa_*\to\bb_{2*}$ are incompatible if and only if there exists some state ensemble $\en:X_1\cup X_2\to\aa_*$ such that 
\begin{equation}
\Prg(\Phi_1,\Phi_2\parallel\pa,\en,\Mov) > \Ppg(\pa,\en,\Mov) \, .
\end{equation}
\end{corollary}

{As in the statement of Theorem \ref{th:main}, also in the above corollary the sets $X_1$, $X_2$ and the measurements $\Mo_1$, $\Mo_2$ are independent of the incompatible channels $\Phi_1$ and $\Phi_2$. Indeed, only the state ensemble $\en$ needs to be arranged to detect incompatibility.}

\begin{proof}[{Proof of Corollary \ref{cor:main}}]
The `if' statement trivially follows from the definition \eqref{eq:Ppost} of $\Ppg(\pa,\en,\Mov)$, so we prove the `only if' part. By Proposition \ref{prop:separation}, there exists a witness $\wit\in\CIW{\aa}{\bb_1,\bb_2}$ such that $\Phiv\in\De{\wit}$. On the other hand, by Theorem \ref{th:main}, we can construct a state ensemble $\en:X_1\cup X_2\to\aa_*$ such that the CIW $\wit_{\pa,\en,\Mov}$ is finer than $\wit$. This means that in \eqref{eq:associated_CIW} we have $\wit_{\pa,\en,\Mov}(\Phiv) < 0$, that is, $\Prg(\Phiv\parallel\pa,\en,\Mov) > \Ppg(\pa,\en,\Mov)$.
\end{proof}

In Corollary \ref{cor:main}, the probability $\Ppg(\pa,\en,\Mov)$ can be calculated analytically or numerically, or at least upper bounded tightly enough, by solving a convex optimization problem. On the other hand, the probability $\Prg(\Phi_1,\Phi_2\parallel\pa,\en,\Mov)$ is assessable by using Alice's classical information, and then performing quantum measurements only on Bob's side. Since no entangled state is shared in the state discrimination protocol, Corollary \ref{cor:main} provides a more practical way to detect incompatibility than schemes based on Bell experiments or steering. In particular, as a fundamental fact, entanglement is not needed to detect incompatibility.

A particular instance of the scheme introduced in this section is the discrimination task with pre- and postmeasurement information described and studied in \cite{CaHeTo18,BaWeWi08,GoWe10,AkKaMa19}. 
In the latter task, Bob is asked to retrieve Alice's label $z\in X_1\cup X_2$ by simply performing a measurement $\No$ on the received state $a_z$, without making any processing of $a_z$ before that. The outcome set of $\No$ is assumed to be the Cartesian product $X_1\times X_2$. When Bob obtains the outcome $(x_1,x_2)$ and Alice communicates him that $z\in X_i$, his guess for $z$ is the value $x_i$. According to the time when Bob is informed about $i$ -- either before or after he performs the measurement -- the choice of $\No$ optimizing the correct guessing probability may be different. Consequently, also in this task Bob's maximal guessing probability may vary according to the pre- or postmeasurement information scenario.

The state discrimination task with pre- or postmeasurement information can be recast into the general scheme described at the beginning of this section by fixing the commutative algebras $\bb_i = \ell^\infty(X_i)$ as Bob's subsystems, identifying the measurement $\No$ with the channel $\Noh : \aa_*\to \ell^1 (X_1\times X_2)$ and letting $\Mo_i : X\to\mathcal{B}_i$ be the projective measurements corresponding to simply reading off the outcome of $\No$. In this way, Theorems 1 and 2 of \cite{CaHeTo19} are particular instances of the above Theorem \ref{th:main} and Corollary \ref{cor:main}.

\section{From measurement to channel incompatibility witnesses}\label{sec:ex_triv}

In this section, we provide examples of a tight channel incompatibility witness $\wit\in\CIW{\aa}{\bb_1,\bb_2}$ for each of the three cases $\bb_1 = \bb_2 = \ell^\infty(X)$ (incompatibility of two measurements), $\bb_1 = \ell^\infty(X)$ and $\bb_2 = \lh$ (incompatibility of a measurement and a channel) and $\bb_1 = \bb_2 = \lh$ (incompatibility of two channels). We always consider the  standard quantum input $\aa=\lh$. 
Moreover, we assume that the cardinality of the outcome set $X$ equals the dimension $d$ of the Hilbert space $\hh$.
Our examples are based on the fact that, by using the next simple observation, the results of \cite{CaHeTo19} immediately yield instances of CIWs also for $\bb_1$ and $\bb_2$ being non-abelian.

\begin{proposition}\label{prop:trivial}
Suppose $\wit\in\CIW{\aa}{\bb_1,\ell^\infty(X)}$ and let $\Po:X\to\bb_2$ be a measurement such that $\Po(x)$ is a nonzero projection of $\bb_2$ for all $x\in X$. Define the map $\wit_\Po :\CC{\aa}{\bb_1,\bb_2}\to\Rb$ as
\begin{equation}
\wit_\Po(\Phiv) = \wit(\Phi_1,\Poh\circ\Phi_2) \qquad \forall\Phiv\in\CC{\aa}{\bb_1,\bb_2} \,.
\end{equation}
Then $\wit_\Po\in\CIW{\aa}{\bb_1,\bb_2}$. Moreover, $\wit_\Po$ is tight if $\wit$ is such.
\end{proposition}

\begin{proof}
Properties (W\ref{it:CIWpositivity}) and (W\ref{it:CIWaffinity}) for $\wit_\Mo$ follow from the analogues properties for $\wit$ and from the fact that $\Phi_1 = {\rm id}\circ\Phi_1$ and $\Poh\circ\Phi_2$ are compatible if $\Phi_1$ and $\Phi_2$ are. In order to prove property (W\ref{it:CIWnegativity}), fix any faithful state $b_0\in\mathcal{S}(\bb_2)$. For all $x\in X$, let $b_{0,x}\in\mathcal{S}(\bb_2)$ be given by $\pair{b_{0,x}}{B} = \pair{b_0}{\Po(x)B\Po(x)}/\pair{b_0}{\Po(x)}$ for all $B\in\bb_2$. Further, define the linear map $\Psi:\ell^1(X)\to\bb_{2*}$ with $\Psi(f) = \sum_{x\in X} f(x)\, b_{0,x}$. Such a map is a channel, since its adjoint $\Psi^*:\bb_2\to\ell^\infty(X)$ is unital and $\ell^\infty(X)$ is abelian. 
Then, it is easy to check that the composition channel $\Poh\circ\Psi$ is the identity map of $\ell^1(X)$, from which it follows that, for any measurement $\Mo:X\to\aa$, we have $\Poh\circ\Psi\circ\Moh = \Moh$. 
In particular, for $\Phi_2 = \Psi\circ\Moh$ we have $\wit_\Po(\Phi_1,\Phi_2) < 0$ if $(\Phi_1,\Moh)\in\De{\wit}$, thus showing property (W\ref{it:CIWnegativity}) for $\wit_\Po$. If instead $\Phi_1$ and $\Moh$ are compatible and $\wit(\Phi_1,\Moh) = 0$, then also $\Phi_1 = {\rm id}\circ\Phi_1$ and $\Phi_2 = \Psi \circ \Moh$ are compatible and $\wit_\Po(\Phi_1,\Phi_2) = 0$, thus implying that $\wit_\Po$ is tight whenever $\wit$ is such.
\end{proof}

The composition channel $\Psi\circ\Moh$ introduced in the previous proof is the {\em measure-and-prepare channel} associated with the measurement $\Mo$ and the family of states $\{b_{0,x}\mid x\in X\}\subset\mathcal{S}(\bb_2)$. Explicitly,
\begin{equation*}
(\Psi\circ\Moh)(a) = \sum_{x\in X} \pair{a}{\Mo(x)}\, b_{0,x} \qquad\forall a\in\aa_* \,.
\end{equation*}
Note that $\Psi\circ\Moh$ is a channel even if the supports of the states $\{b_{0,x}\mid x\in X\}$ are not orthogonal.

We start constructing our examples by recalling a family of inequivalent tight witnesses $\wit_\bmu\in\CIW{\lh}{\ell^\infty(X),\ell^\infty(X)}$ which was described in \cite[Thm.~3]{CaHeTo19}. 
This family is constructed by fixing two mutually unbiased bases $\{e_x \mid x\in X\}$ and $\{f_x \mid x\in X\}$ of $\hh$, and it depends on the direction of a two dimensional vector $\bmu\in\Rb^2$. 
Here we recall only the following example, which corresponds to the equally weighted choice $\bmu = (1,1)$ for the parameter $\bmu$:
\begin{equation}\label{eq:meas-meas}
\wit_{\rm mm}(\Moh,\Noh) = \frac{1}{2d}\bigg\{\sqrt{d}(\sqrt{d}+1) - \sum_{x\in X} \left[\ip{e_x}{\Mo(x)e_x} + \ip{f_x}{\No(x)f_x}\right]\bigg\} \,.
\end{equation}
In the previous formula, $\ip{\cdot}{\cdot}$ is the inner product of the Hilbert space $\hh$. The measurement-measurement incompatibility witness \eqref{eq:meas-meas} gives zero when evaluated on the compatible pair of quantum measurements
\begin{equation}\label{eq:M0N0}
\Mo_0(x) = \gamma(d)\kb{e_x}{e_x} + (1-\gamma(d))\,\frac{\id}{d}\,,\qquad\quad \No_0(x) = \gamma(d)\kb{f_x}{f_x} + (1-\gamma(d))\,\frac{\id}{d} \,,
\end{equation}
where $\id$ is the identity operator on $\hh$ and $\gamma(d)$ is the real constant
\begin{equation}\label{eq:gamma(d)}
\gamma(d) = \frac{\sqrt{d}+2}{2(\sqrt{d}+1)}\,.
\end{equation}
This implies that the measurements 
\begin{equation*}
\Mo(x) = \gamma\kb{e_x}{e_x} + (1-\gamma)\,\frac{\id}{d}\,,\qquad\quad \No(x) = \gamma\kb{f_x}{f_x} + (1-\gamma)\,\frac{\id}{d} \,,
\end{equation*}
are incompatible if and only if $\gamma(d) < \gamma \leq 1$, a result that was earlier obtained in \cite{UoLuMoHe16,DeSkFrBr19} by using different methods.

The previous measurement-measurement witness can be immediately turned into a tight witness $\wit_{\rm mc}\in\CIW{\lh}{\ell^\infty(X),\lh}$ by means of Proposition \ref{prop:trivial}. 
Indeed, it is enough to fix another orthonormal basis $\{h_x \mid x\in X\}$, set $\Po(x) = \kb{h_x}{h_x}$ and define
\begin{equation}\label{eq:meas-chan_bad}
\begin{aligned}
& \wit_{\rm mc}(\Moh,\Lambda) = (\wit_{\rm mm})_\Po(\Moh,\Lambda) = \wit_{\rm mm}(\Moh,\Poh\circ\Lambda) \\
& \qquad = \frac{1}{2d}\bigg\{\sqrt{d}(\sqrt{d}+1) - \sum_{x\in X} \left[\ip{e_x}{\Mo(x)e_x} + \ip{h_x}{\Lambda(\kb{f_x}{f_x})\,h_x}\right]\bigg\} \,.
\end{aligned}
\end{equation}
We have $\wit_{\rm mc}(\Moh_0,\Lambda_{\No_0}) = 0$ for the compatible pair $(\Moh_0,\Lambda_{\No_0})$, in which $\Moh_0$ is given by \eqref{eq:M0N0} and $\Lambda_{\No_0}$ is the measure-and-prepare quantum channel
\begin{equation}\label{eq:meas-prep}
\Lambda_{\No_0}(a) = \sum_{x\in X}\tr{a\No_0(x)} \kb{h_x}{h_x} \qquad\forall a\in\loneh
\end{equation}
with $\No_0$ still given by \eqref{eq:M0N0}. {Here, ${\rm tr}$ denotes the trace of $\hh$.}

In order to find an example of a tight witness $\wit_{\rm cc}\in\CIW{\lh}{\lh,\lh}$, we can still proceed along the same lines as previously. 
Specifically, we can use the witness \eqref{eq:M0N0} and any two bases $\{g_x \mid x\in X\}$ and $\{h_x \mid x\in X\}$ of $\hh$ in order to construct $\wit_{\rm cc}$ by means of Proposition \ref{prop:trivial}. In this way, dropping the irrelevant factor $1/(2d)$, the resulting witness is
\begin{equation}\label{eq:chan-chan_bad}
\wit_{\rm cc}(\Theta,\Lambda) = \sqrt{d}(\sqrt{d}+1) - \sum_{x\in X} \left[\ip{g_x}{\Theta(\kb{e_x}{e_x})\,g_x} + \ip{h_x}{\Lambda(\kb{f_x}{f_x})\,h_x}\right]
\end{equation}
for all $(\Theta,\Lambda)\in\CC{\lh}{\lh,\lh}$.

\section{Incompatibility witness related to approximate cloning}\label{sec:ex_clon}
\label{app:chan-chan}

As we have seen, the measurement-channel and the channel-channel incompatibility witnesses $\wit_{\rm mc}$ and $\wit_{\rm cc}$ derived in the previous section are adaptations of the measurement-measurement witness $\wit_{\rm mm}$ found in \cite{CaHeTo19} and constructed by means of two mutually unbiased bases. 
Here we show that, by using a different method, another tight witness $\witt_{\rm cc}\in\CIW{\lh}{\lh,\lh}$ can also be derived by fixing only one arbitrary orthonormal basis $\{e_x \mid x\in X\}$ of $\hh$ and setting
\begin{subequations}\label{eq:chan-chan}
\begin{equation}\label{eq:chan-chan_a}
\witt_{\rm cc}(\Theta,\Lambda) = d(d+1) - \sum_{x,y\in X}\ip{e_x}{(\Theta+\Lambda)(\kb{e_x}{e_y})\,e_y} \,.
\end{equation}
Actually, the dependence of $\witt_{\rm cc}$ on the choice of the basis of $\hh$ is not relevant. Indeed, \eqref{eq:chan-chan_a} can be rewritten in a basis independent form by using the trace ${\rm Tr}$ of the linear space $\lh$, so that
\begin{equation}\label{eq:chan-chan_b}
\witt_{\rm cc}(\Theta,\Lambda) = d(d+1) - {\rm Tr}[\Theta+\Lambda] \,.
\end{equation}
\end{subequations}

For the witness $\witt_{\rm cc}$, we have $\witt_{\rm cc}(\Theta_0,\Lambda_0) = 0$ when $\Theta_0$ and $\Lambda_0$ are the two margins of the optimal approximate cloning channel found in \cite{Werner98,KeWe99}, i.e., the depolarizing channels
\begin{equation}\label{eq:W_chan}
\Theta_0(a) = \Lambda_0(a) = \gamma(d^2) a + (1-\gamma(d^2))\tr{a}\,\frac{\id}{d}
\end{equation}
with $\gamma(d^2)$ defined by \eqref{eq:gamma(d)}. 

One can show with some calculation that the witnesses $\wit_{\rm cc}$ and $\witt_{\rm cc}$ are detection inequivalent, since inserting $\Theta_0$ and $\Lambda_0$ into \eqref{eq:chan-chan_bad} yields
\begin{equation*}
\wit_{\rm cc}(\Theta_0,\Lambda_0) = (\sqrt{d}+2)(\sqrt{d}-1) + \gamma(d^2) \bigg[2-\sum_{x\in X} (\abs{\ip{e_x}{g_x}}^2 + \abs{\ip{f_x}{h_x}}^2)\bigg] \,,
\end{equation*}
which is strictly positive for all $d\geq 2$ and any choice of the bases $\{g_x \mid x\in X\}$ and $\{h_x \mid x\in X\}$. Thus, for suitably small $\varepsilon>0$, the CIW $\wit_{\rm cc}$ does not detect the incompatible channels $\Theta = (1+\varepsilon)\Theta_0 - \varepsilon\tr{\cdot}\id/d$ and $\Lambda = (1+\varepsilon)\Lambda_0 - \varepsilon\tr{\cdot}\id/d$, which instead are detected by $\witt_{\rm cc}$.

In a similar way, if we insert the compatible measure-and-prepare channels
\begin{equation}
\Theta_{\Mo_0}(a) = \sum_{x\in X}\tr{a\Mo_0(x)} \kb{g_x}{g_x} \,,\qquad \Lambda_{\No_0}(a) = \sum_{x\in X}\tr{a\No_0(x)} \kb{h_x}{h_x}
\end{equation}
into \eqref{eq:chan-chan}, we obtain
\begin{equation*}
\witt_{\rm cc}(\Theta_{\Mo_0},\Lambda_{\No_0}) = (d+2)(d-1) + \gamma(d) \bigg[2-\sum_{x\in X} (\abs{\ip{e_x}{g_x}}^2 + \abs{\ip{f_x}{h_x}}^2)\bigg] \,,
\end{equation*}
which is strictly positive for all $d\geq 2$ and any bases $\{g_x \mid x\in X\}$ and $\{h_x \mid x\in X\}$. Since on the other hand $\wit_{\rm cc}(\Theta_{\Mo_0},\Lambda_{\No_0}) = 0$, a similar reasoning as in the previous paragraph yields that $\De{\wit_{\rm cc}}\not\subseteq\De{\witt_{\rm cc}}$. Thus, neither $\wit_{\rm cc}$ is finer than $\witt_{\rm cc}$, nor $\witt_{\rm cc}$ is finer than $\wit_{\rm cc}$, thus proving that the two witnesses $\wit_{\rm cc}$ and $\witt_{\rm cc}$ are genuinely diverse.

The rest of this section is devoted to the proof that the map $\witt_{\rm cc}$ defined in \eqref{eq:chan-chan} is a tight CIW, and that $\witt_{\rm cc}(\Theta_0,\Lambda_0) = 0$ when $\Theta_0$ and $\Lambda_0$ are the compatible channels defined in \eqref{eq:W_chan}.

For any pair of channels $(\Theta,\Lambda)\in \CC{\lh}{\lh ,\lh}$, let
\begin{equation*}
\witt_0 (\Theta,\Lambda) =\sum_{i,j=1}^d  \ip{e_i}{\left(\Theta+\Lambda\right) (\kb{e_i}{e_j})e_j}
\end{equation*}
be the linear part of the witness \eqref{eq:chan-chan}.  
By denoting
\begin{equation*}
\omega = \frac{1}{\sqrt{d}} \sum_{i=1}^d e_i \otimes e_i 
\end{equation*}
the maximally entangled state associated with the given basis, the linear functional $\witt_0$ can be rewritten as
\begin{equation*}
\witt_0 (\Theta,\Lambda)= d^2 \left\{\tr{\kb{\omega}{\omega}\left(\Theta^* \otimes {\rm id}^*\right)(\kb{\omega}{\omega})} + \tr{\kb{\omega}{\omega}\left(\Lambda^* \otimes{\rm id}^*\right)(\kb{\omega}{\omega})}\right\} \,,
\end{equation*}
where ${\rm id}:\loneh\to\loneh$ is the identity channel.

Now, suppose $\Theta$ and $\Lambda$ are compatible, and let $\Phi$ be a joint channel for them. 
Moreover, denote by $F \colon \hh \otimes \hh \to \hh\otimes \hh$ the flip operator $F(u\otimes v)= v\otimes u$. Then, using the marginality conditions $\Phi^*(A\otimes\id) = \Theta^*$ and $\Phi^*(\id\otimes B) = \Lambda^*$ together with the relation $A\otimes\id = F(\id\otimes A)F$, we have
\begin{equation}\label{eq:chanchancomp}
\begin{aligned}
\witt_0 (\Theta,\Lambda) = &\, d^2 \big\{\tr{\kb{\omega}{\omega}\left( \Phi^\ast \otimes {\rm id}^*\right) \left((F\otimes \id)(\id\otimes \kb{\omega}{\omega})(F\otimes \id)\right)} \\
& \, + \tr{\kb{\omega}{\omega}\left( \Phi^\ast \otimes {\rm id}^* \right) (\id \otimes \kb{\omega}{\omega})}\big\} \\
= &\, d^2\tr{(\Phi \otimes {\rm id})(\kb{\omega}{\omega})\, E} \,,
\end{aligned}
\end{equation}
where $E$ is the selfadjoint positive operator
\begin{equation*}
E= (F\otimes \id) (\id \otimes \kb{\omega}{\omega})(F\otimes \id) + \id \otimes \kb{\omega}{\omega} \,.
\end{equation*}
Hence, for any compatible pair $(\Theta,\Lambda)$, we have the following upper bound for \eqref{eq:chanchancomp}
\begin{equation}
\label{eq:chanchanbound1}
\witt_0 (\Theta,\Lambda)\leq d^2\lambda_{\rm max}(E) \,,
\end{equation}
where $\lambda_{\rm max}(E)$ is the maximal eigenvalue of $E$.

We now evaluate $\lambda_{\rm max}(E)$ by finding the eigenspace decomposition of $E$. To this aim, we introduce the two operators $S_\pm = \frac{1}{2} \left( \id \pm F\right)$,
which are the orthogonal projections onto the symmetric and antisymmetric subspaces of $\hh\otimes \hh$, respectively. Since 
\begin{equation}\label{eq:E_action}
\begin{aligned}
& E\left(
(e_l \otimes e_m \pm e_m \otimes e_l)\otimes u
\right) = \\
& = \frac{1}{d} \left[
\sum_{i=1}^d \ip{e_m}{u} \left(
e_l \otimes e_i \pm e_i \otimes e_l
\right)\otimes e_i
\pm
\sum_{i=1}^d \ip{e_l}{u} \left(
e_m \otimes e_i \pm e_i \otimes e_m
\right)\otimes e_i
\right]\,,
\end{aligned}
\end{equation}
we conclude that
\begin{equation*}
E\left(
S_\pm(\hh\otimes \hh)\otimes \hh
\right)=\left\{
\sum_{i=1}^d  (v\otimes  e_i \pm e_i \otimes v)\otimes e_i \,\mid \, v \in \hh
\right\}\,.
\end{equation*}
Since $E$ commutes with both projections $S_+\otimes\id$ and $S_-\otimes\id$ and $S_+ + S_- = \id$, we have the orthogonal decomposition
\begin{equation*}
E(\hh\otimes \hh \otimes \hh) = E\left(
S_+(\hh\otimes \hh)\otimes \hh
\right)\oplus E\left(
S_-(\hh\otimes \hh)\otimes \hh
\right) \, .
\end{equation*}
Now we show that $E\left(
S_+(\hh\otimes \hh)\otimes \hh
\right)$ and $E\left(
S_-(\hh\otimes \hh)\otimes \hh
\right)$ are the eigenspaces corresponding to the only two nonzero eigenvalues $\lambda_\pm(E) = (d\pm 1)/d$ of $E$. Indeed, another application of \eqref{eq:E_action} yields
\begin{equation*}
E\left(
\sum_{i=1}^d  (v\otimes  e_i \pm e_i \otimes v)\otimes e_i 
\right) = \frac{d\pm 1}{d} \sum_{i=1}^d  (v\otimes  e_i \pm e_i \otimes v)\otimes e_i \,.
\end{equation*}
We thus conclude that $\lambda_{\rm max}(E) = (d+1)/d$, hence, for any compatible pair $(\Theta,\Lambda)$, by \eqref{eq:chanchanbound1} we have
\begin{equation}\label{eq:chanchanbound2}
\witt_0 (\Theta,\Lambda)\leq d(d+1) \,.
\end{equation}

On the other hand, since for the identity channel we have
\begin{equation*}
\witt_0 (\mathrm{id}, \mathrm{id})=2d^2 \abs{\ip{\omega}{\omega}}^2 = 2d^2 > d(d+1) \,,
\end{equation*}
it follows that \eqref{eq:chan-chan} defines a CIW.

Finally, as we already noticed, by \cite{Werner98,KeWe99} the two depolarizing channels $\Theta_0$ and $\Lambda_0$ defined in \eqref{eq:W_chan} are compatible, and an easy calculation yields $\witt_0 (\Theta_0,\Lambda_0)= d(d+1)$. 
The bound \eqref{eq:chanchanbound2} is thus attained on $\cCC{\lh}{\lh,\lh}$, hence the witness \eqref{eq:chan-chan} is tight.

\section{Discussion}

We have proved that incompatibility can always be detected by means of a state discrimination protocol. 
We have done it for systems described by arbitrary finite dimensional von Neumann algebras, thus encompassing all possible hybrid quantum-classical cases. Our approach was based on the notion of channel incompatibility witness and its connection with a state discrimination task with intermediate partial information. Once we established this connection in Theorem \ref{th:main}, the main result in Corollary \ref{cor:main} easily followed from standard separation results for convex compact sets. We pointed out that all incompatible pairs of channels can be detected by tuning only the state ensemble on Alice's side, while Bob can keep his measurements fixed to this purpose.  

The essential point in the presented formalism is that the set of all compatible channels is a convex compact subset of all pairs of channels. 
In fact, a similar mathematical technique works for any binary relation $\mathcal{R}\subset\CC{\aa}{\bb_1,\bb_2}$ that is convex and compact.
The state discrimination protocol is hence useful to detect also other resources, mathematically described as subsets of $\CC{\aa}{\bb_1,\bb_2}$ with convex compact complements.

We finally provided four examples of channel incompatibility witnesses with standard quantum input $\aa=\lh$. The first example \eqref{eq:meas-meas} applies to measurement-measurement incompatibility and was taken from \cite{CaHeTo19}, while the second \eqref{eq:meas-chan_bad} and the third \eqref{eq:chan-chan_bad} are adaptations of the former one. The last example of channel-channel incompatibility witness \eqref{eq:chan-chan} is unrelated to the measurement-measurement case. It would be interesting to develop it into a whole family of inequivalent witnesses in analogy with the results of \cite{CaHeTo19}. We also point out that the mea\-sure\-ment-channel case deserves further study, as the only presented example relies upon the measurement-measurement case. We defer more detailed investigations on these topics to future work.

\section*{Acknowledgements}

TH  acknowledges financial support from the Academy of Finland via the Centre of Excellence program (Project no.~312058) as well as Project no.~287750.
TM acknowledges financial support from JSPS (KAKENHI Grant Number 15K04998).

\section*{Note added}

During the preparation of the manuscript we became aware of recent related works by Uola, Kraft and Abbott \cite{UoKrAb19} and by Mori \cite{Mori19}.


\begin{thebibliography}{99}

\bibitem{HeMiZi16}
T.~Heinosaari, T.~Miyadera, and M.~Ziman,
\newblock ``An invitation to quantum incompatibility'',
\newblock {\em J. Phys. A: Math. Theor.} {\bf 49}(12), 123001 (2016).

\bibitem{BrDVEkFuMaSm98}
D.~Bru\ss, D.P.~DiVincenzo, A.~Ekert, C.A.~Fuchs, C.~Macchiavello, and J.A.~Smolin,
\newblock ``Optimal universal and state-dependent quantum cloning'',
\newblock {\em Phys. Rev. A} {\bf 57}(4), 2368--2378 (1998).

\bibitem{Werner98}
R.F.~Werner,
\newblock ``Optimal cloning of pure states'',
\newblock {\em Phys. Rev. A} {\bf 58}(3), 1827--1832 (1998).

\bibitem{KeWe99}
M.~Keyl and R.F.~Werner,
\newblock ``Optimal cloning of pure states, testing single clones'',
\newblock {\em J. Math. Phys.} {\bf 40}(7), 3283--3299 (1999).

\bibitem{Cerf00bis}
N.J.~Cerf,
\newblock ``Asymmetric quantum cloning in any dimension'',
\newblock {\em J.~Mod.~Optic.} {\bf 47}(2-3), 187--209 (2000).

\bibitem{BrBuHi01}
S.L.~Braunstein, V.~Bu\v{z}ek, and M.~Hillery,
\newblock ``Quantum-information distributors: Quantum network for symmetric and asymmetric cloning in arbitrary dimension and continuous limit'',
\newblock {\em Phys. Rev. A} {\bf 63}(5), 0523131 (2001).

\bibitem{Hashagen17}
A.-L.~Hashagen,
\newblock ``Universal asymmetric quantum cloning revisited'',
\newblock {\em Quantum Inf. Comput.} {\bf 17}(9-10), 747--778 (2017).

\bibitem{HeMi17}
T.~Heinosaari and T.~Miyadera,
\newblock ``Incompatibility of quantum channels'',
\newblock {\em J. Phys. A: Math. Theor.} {\bf 50}(13), 135302 (2017).

\bibitem{CaHeTo19}
C.~Carmeli, T.~Heinosaari, and A.~Toigo,
\newblock ``Quantum incompatibility witnesses'',
\newblock {\em Phys. Rev. Lett.} {\bf 122}(13), 130402 (2019).

\bibitem{UoKrShYuGu19}
R.~Uola, T.~Kraft, J.~Shang, X.-D.~Yu, and O.~G\"uhne,
\newblock ``Quantifying quantum resources with conic programming'',
\newblock {\em Phys. Rev. Lett.} {\bf 122}(13), 130404 (2019).

\bibitem{SkSuCa19}
P.~Skrzypczyk, \ifmmode \check{S}\else I.~\v{S}\fi{}upi\ifmmode~\acute{c}\else\'{c}\fi{}, and D.~Cavalcanti,
\newblock ``All sets of incompatible measurements give an advantage in quantum state discrimination'',
\newblock {\em Phys. Rev. Lett.} {\bf 122}(13), 130403 (2019).

\bibitem{CaHeTo18}
C.~Carmeli, T.~Heinosaari, and A.~Toigo,
\newblock ``State discrimination with postmeasurement information and incompatibility of quantum measurements'',
\newblock {\em Phys. Rev. A} {\bf 98}(1), 012126 (2018).

\bibitem{Keyl02}
M.~Keyl,
\newblock ``Fundamentals of quantum information theory'',
\newblock {\em Phys. Rep.} {\bf 369}(5), 431--548 (2002).

\bibitem{Jencova12}
A.~Jen{\v{c}}ov\'a,
\newblock ``Generalized channels: {C}hannels for convex subsets of the state space'',
\newblock {\em J. Math. Phys.} {\bf 53}(1), 012201 (2012).

\bibitem{HeMi13}
T.~Heinosaari and T.~Miyadera,
\newblock ``Qualitative noise-disturbance relation for quantum measurements'',
\newblock {\em Phys. Rev. A} {\bf 88}(4), 042117 (2013).

\bibitem{HeReRyZi18}
T.~Heinosaari, D.~Reitzner, T.~Ryb\'{a}r, and M.~Ziman,
\newblock ``Incompatibility of unbiased qubit observables and {P}auli channels'',
\newblock {\em Phys. Rev. A} {\bf 97}(2), 022112 (2018).

\bibitem{CaHeMiTo19}
C.~Carmeli, T.~Heinosaari, T.~Miyadera, and A.~Toigo,
\newblock ``Noise-disturbance relation and the {G}alois connection of quantum measurements'',
\newblock {\em Found. Phys.} {\bf 49}(6), 492--505 (2019).

\bibitem{HaMi19}
I.~Hamamura and T.~Miyadera,
\newblock ``Relation between state-distinction power and disturbance in quantum measurements'',
\newblock {\em J. Math. Phys.} {\bf 60}(8), 082103 (2019).

\bibitem{SiSt92}
M.~Singer and W.~Stulpe,
\newblock ``Phase-space representations of general statistical physical theories'',
\newblock {\em J. Math. Phys.} {\bf 33}(1), 131--142 (1992).

\bibitem{MLQT12}
T.~Heinosaari and M.~Ziman,
\newblock {\em The {M}athematical {L}anguage of {Q}uantum {T}heory}
\newblock (Cambridge University Press, Cambridge, 2012).

\bibitem{EIT06}
T.M.~Cover and J.A.~Thomas,
\newblock {\em Elements of information theory}, second edition
\newblock (Wiley-Interscience [John Wiley \& Sons], Hoboken, NJ, 2006).

\bibitem{HaHePe14}
E.~Haapasalo, T.~Heinosaari, and J.-P.~Pellonp\"a\"a,
\newblock ``When do pieces determine the whole? Extreme marginals of a completely positive map'',
\newblock {\em Rev. Math. Phys.} {\bf 26}(2), 1450002 (2014).

\bibitem{Haapasalo15}
E.~Haapasalo,
\newblock ``Robustness of incompatibility for quantum devices'',
\newblock {\em J. Phys. A: Math. Theor.} {\bf 48}(25), 255303 (2015).

\bibitem{Plavala17}
M.~Pl\'avala,
\newblock ``Conditions for the compatibility of channels in general probabilistic theory and their connection to steering and {B}ell nonlocality'',
\newblock {\em Phys. Rev. A} {\bf 96}(5), 052127 (2017).

\bibitem{Kuramochi18a}
Y.~Kuramochi,
\newblock ``Quantum incompatibility of channels with general outcome operator algebras'',
\newblock {\em J. Math. Phys.} {\bf 59}(4), 042203 (2018).

\bibitem{TOA1}
M.~Takesaki,
\newblock {\em Theory of operator algebras. {I}}, volume 124 of {\em Encyclopaedia of Mathematical Sciences}
\newblock (Springer-Verlag, Berlin, 2002).

\bibitem{Plavala16}
M.~Pl\'avala,
\newblock ``All measurements in a probabilistic theory are compatible if and only if the state space is a simplex'',
\newblock {\em Phys. Rev. A} {\bf 94}(4), 042108 (2016).

\bibitem{CBMOA03}
V.~Paulsen,
\newblock {\em Completely bounded maps and operator algebras}
\newblock (Cambridge University Press, Cambridge, 2003).

\bibitem{BaBaLeWi07}
H.~Barnum, J.~Barrett, M.~Leifer, and A.~Wilce,
\newblock ``Generalized no-broadcasting theorem'',
\newblock {\em Phys. Rev. Lett.} {\bf 99}(24), 240501 (2007).

\bibitem{KaLuLu15}
K.~Kaniowski, K.~Lubnauer, and A.~{\L}uczak,
\newblock ``Cloning and broadcasting in operator algebras'',
\newblock {\em Q. J. Math.} {\bf 66}(1), 191--212 (2015).

\bibitem{Jencova18}
A.~Jen{\v{c}}ov\'a,
\newblock ``Incompatible measurements in a class of general probabilistic theories'',
\newblock {\em Phys. Rev. A} {\bf 98}(1), 012133 (2018).

\bibitem{BlNe18}
A.~Bluhm and I.~Nechita,
\newblock ``Compatibility of quantum measurements and inclusion constants for the matrix jewel'',
\newblock arXiv:1809.04514 [quant-ph].

\bibitem{CA70}
R.T.~Rockafellar,
\newblock {\em Convex {A}nalysis}
\newblock (Princeton University Press, 1970).

\bibitem{HeMaWo13}
T.~Heinosaari, L.~Mazzarella, and M.M.~Wolf,
\newblock ``Quantum tomography under prior information'',
\newblock {\em Comm. Math. Phys.} {\bf 318}(2), 355--374 (2013).

\bibitem{BaWeWi08}
M.A.~Ballester, S.~Wehner, and A.~Winter,
\newblock ``State discrimination with post-measurement information'',
\newblock {\em IEEE Trans. Inf. Theory} {\bf 54}(9), 4183--4198 (2008).

\bibitem{GoWe10}
D.~Gopal and S.~Wehner,
\newblock ``Using postmeasurement information in state discrimination'',
\newblock {\em Phys. Rev. A} {\bf 82}(2), 022326 (2010).

\bibitem{AkKaMa19}
S.~Akibue, G.~Kato, and N.~Marumo,
\newblock ``Perfect discrimination of non-orthogonal quantum states with posterior classical partial information'',
\newblock {\em Phys. Rev. A} {\bf 99}(2), 020102(R) (2019).

\bibitem{UoLuMoHe16}
R.~Uola, K.~Luoma, T.~Moroder, and T.~Heinosaari,
\newblock ``Adaptive strategy for joint measurements'',
\newblock {\em Phys. Rev. A} {\bf 94}(2), 022109 (2016).

\bibitem{DeSkFrBr19}
S.~Designolle, P.~Skrzypczyk, F.~Fr\"owis, and N.~Brunner,
\newblock ``Quantifying measurement incompatibility of mutually unbiased bases'',
\newblock {\em Phys. Rev. Lett.} {\bf 122}(5), 050402 (2019).

\bibitem{UoKrAb19}
R.~Uola, T.~Kraft, and A.A.~Abbott,
\newblock ``Quantification of quantum dynamics with input-output games'',
\newblock arXiv:1906.09206 [quant-ph].

\bibitem{Mori19}
J.~Mori,
\newblock ``Operational characterization of incompatibility of quantum channels with quantum state discrimination'',
\newblock arXiv:1906.09859 [quant-ph].

\end{thebibliography}
\end{document}